\documentclass[conference]{IEEEtran}
\usepackage{amsmath,amssymb,amsthm}
\usepackage{cite}
\usepackage{hyperref}
\usepackage{graphicx}
\usepackage{algorithmic}
\usepackage{textcomp}
\usepackage{booktabs}
\usepackage{multirow}
\usepackage{pgfplots}
\pgfplotsset{compat=1.18}
\usepackage{tikz}
\usetikzlibrary{arrows.meta, positioning, calc, decorations.pathreplacing}

\newtheorem{theorem}{Theorem}
\newtheorem{lemma}[theorem]{Lemma}
\newtheorem{corollary}[theorem]{Corollary}
\newtheorem{definition}{Definition}
\newtheorem{proposition}{Proposition}
\newtheorem{remark}{Remark}
\newtheorem{axiom}{Axiom}

\hypersetup{
    colorlinks=true,
    linkcolor=black,
    citecolor=black,
    urlcolor=blue
}

\title{A Formal Framework for Critical-Mass Collapse in Online Multiplayer Games}

\author{
    \IEEEauthorblockN{Ahmed Sheta, FRSA}
    \IEEEauthorblockA{
        Department of Computer Science\\
        Georgia Institute of Technology
    }
}

\begin{document}
\maketitle

\begin{abstract}
Online multiplayer games are population-dependent systems whose playability depends on the continued presence of an active player base. We propose a formal framework for reasoning about viability collapse in such systems under explicit scope conditions. The framework introduces a conditional Critical Mass Threshold $\Phi$, below which queue times, match quality, or role balance render a game operationally non-viable under a fixed operational profile; an uninhabited runtime taxonomy spanning pre-launch and post-decline states; and a Nostalgia Inversion Point $\psi$, at which cultural memory exceeds active participation. We model post-peak decline using a threshold-sensitive hazard model and show how games in the modeled class can cross below viability under finite official-service horizons or bounded novelty under continuing exposure. Case studies based on public concurrent-player data are used illustratively rather than as formal validation. The contribution of the paper is not a universal law, but a formal vocabulary, a collapse model, and an empirical agenda for studying online game decline, preservation risk, and uninhabited virtual worlds.
\end{abstract}


%-------------------------------------------------------
\section{Introduction}
\label{sec:intro}

Online multiplayer games are unusual digital artifacts because their value is not contained solely in executable code or audiovisual assets. A multiplayer title remains playable only while enough players are present to sustain its core interaction loop through matchmaking, role balance, social availability, and perceived world activity. Preserving software binaries is therefore not equivalent to preserving the lived system. A game can remain technically runnable while becoming socially non-viable.

This paper studies that transition. We propose a formal framework for reasoning about how online multiplayer games move from active population states to operational decline, and how post-decline runtime states relate to questions of cultural memory and preservation.

The motivating phenomenon is widespread. The Video Game History Foundation's 2023 report found that only 13\% of classic video games released in the United States remained commercially available, implying that 87\% were out of print \cite{b_vghf}. The Digital Preservation Coalition reclassified shut-down video games as ``Practically Extinct'' in 2023 \cite{b_dpc}. Preservation campaigns such as the European Citizens' Initiative \textit{Stop Destroying Videogames} further underline the scope of the problem \cite{b_stopdestroyinggames}. Prior work on player behavior, network effects, churn, and online-community collapse has typically treated one component of the lifecycle at a time rather than integrating adoption, post-peak decline, matchmaking viability, and post-service states into a single framework.

We use the term \textit{digital thanatology} as a descriptive label for this line of inquiry, but the paper's contribution is narrower and more technical. We do not claim a universal law of decline. Instead, we present a scope-limited formal framework whose claims are conditional on explicit assumptions and whose empirical validation remains open.

Our contributions are fourfold. First, we formalize a conditional \textit{Critical Mass Threshold} $\Phi$ as an operational viability boundary under a fixed matchmaking profile. Second, we present a threshold-sensitive population-decline model that captures self-reinforcing sub-critical collapse under explicit assumptions. Third, we define an uninhabited runtime taxonomy spanning pre-launch and post-decline states. Fourth, we link decline to preservation questions through the concepts of a Nostalgia Inversion Point $\psi$ and a Preservation Window $\mathcal{W}$.

\textit{What is new and what is not.} The conditional viability formulation of $\Phi$, the integrated uninhabited runtime taxonomy, the Ghost Machinery state $\Omega_0$, and the explicit treatment of nostalgia and preservation as lifecycle concepts are original to this paper. The Weibull engagement literature \cite{b_bauckhage2012}, diffusion-model literature \cite{b_bass1969}, network-effects literature \cite{b_metcalfe2013,b_shankar2003}, and collapse analogies from online communities and ecology are adopted as supporting foundations rather than presented as originating here.

The framework is formal where explicit modeling buys clarity and deliberately modest where evidence is currently limited. The case studies in this paper are illustrative rather than validating. Titles with strong user-generated content ecosystems, private-server continuation, bot substitution, or major relaunches may fall partly or wholly outside the modeled class.

The paper proceeds as follows. Section~\ref{sec:foundations} establishes definitions, axioms, and scope limitations. Section~\ref{sec:population} develops the population dynamics model. Section~\ref{sec:ghost} introduces the uninhabited runtime taxonomy. Section~\ref{sec:inevitability} derives a conditional critical-mass crossing result. Section~\ref{sec:nostalgia} formalizes the Nostalgia Inversion Point. Section~\ref{sec:paradox} discusses preservation tension and the Preservation Window. Section~\ref{sec:empirical} provides illustrative case studies, and Section~\ref{sec:conclusion} discusses implications and limitations.

%-------------------------------------------------------
\section{Foundational Definitions and Axioms}
\label{sec:foundations}

We establish the formal objects of the theory.

\begin{definition}[Online Game System]
An \textbf{Online Game System} is a tuple $\mathcal{G} = (S, P, C, U, \mathcal{M})$ where:
\begin{itemize}
    \item $S$ is the \textbf{server infrastructure}: the set of computational resources hosting the game state;
    \item $P: \mathbb{R}_{\geq 0} \to \mathbb{Z}_{\geq 0}$ is the \textbf{population function}, mapping time $t$ to the count of active players;
    \item $C: \mathbb{R}_{\geq 0} \to \mathbb{R}_{\geq 0}$ is the \textbf{cumulative content function}, representing the total volume of designed interactive content available at time $t$;
    \item $U: \mathbb{R}_{\geq 0} \to \mathbb{R}_{\geq 0}$ is the \textbf{network utility function}, representing the social and interactive value derived from population size;
    \item $\mathcal{M}: \mathbb{R}_{\geq 0} \to \mathbb{R}_{\geq 0}$ is the \textbf{cultural memory function}, representing the number of individuals who retain experiential knowledge of $\mathcal{G}$ at time $t$.
\end{itemize}
\end{definition}

\begin{definition}[Active Player]
A player $p_i$ is \textbf{active} at time $t$ if $p_i$ has initiated at least one interactive session with $\mathcal{G}$ within the interval $[t - \delta, t]$, where $\delta$ is a game-specific \textbf{activity window} (heuristically, 30 days for persistent-world games, 7 days for session-based games; alternative values yield equivalent qualitative structure). This aligns with the Monthly Active User (MAU) and Weekly Active User (WAU) metrics standard in the games industry.
\end{definition}

\begin{definition}[Content Production Rate]
The \textbf{content production rate} is $c(t) = \frac{dC}{dt}$, representing the instantaneous rate of new content introduction via updates, expansions, patches, and community-generated modifications.
\end{definition}

We now state the axioms governing the system. These axioms are empirically motivated; justifications from the literature follow each statement.

\begin{axiom}[Finite Content Production]
\label{ax:finite}
For any game $\mathcal{G}$ with finite development resources, there exists a bound $c_{\max} < \infty$ such that $c(t) \leq c_{\max}$ for all $t$, and $\lim_{t \to \infty} c(t) = 0$.
\end{axiom}

\noindent\textit{Justification.} Studios work with finite budgets and finite headcount, and the rate at which they can produce new content is bounded by those constraints. Even titles with unusually long content pipelines, such as \textit{World of Warcraft}, exhibit gradually lengthening intervals between expansions as the years accumulate. Content production does not grow forever. It plateaus, then slows, and the axiom is meant to capture that plateau rather than to predict its exact timing.

\begin{axiom}[Monotonic Exposure]
\label{ax:exposure}
The cumulative player experience function $E(t) = \int_0^t P(\tau) \cdot \bar{h}(\tau) \, d\tau$, where $\bar{h}(\tau)$ is the mean session duration, is monotonically non-decreasing.
\end{axiom}

\noindent\textit{Justification.} Both $P(t)$ and $\bar{h}(t)$ are non-negative quantities, so their integral grows or holds steady but cannot shrink. Players accumulate experience in one direction. There is no mechanism in the system for a player to forget content they have already seen.

\begin{axiom}[Network-Dependent Utility]
\label{ax:network}
For a fixed operational profile, the network utility function takes the form $U(P)=\alpha P^{\beta}$ for constants $\alpha>0$ and $\beta>0$. The exponent $\beta$ is treated as an empirical parameter rather than a universal constant. In particular, $U(P) \to 0$ as $P \to 0$.
\end{axiom}

\noindent\textit{Justification.} Metcalfe-style quadratic scaling is a useful limiting case, but empirical work on game and network-product markets suggests that network effects vary over the product lifecycle and need not be exactly quadratic. We therefore use the power-law form as a flexible family that includes the classical case $\beta=2$ while allowing weaker or stronger effective scaling under different designs \cite{b_metcalfe2013,b_liu2015}.

\begin{axiom}[Official-Service Infrastructure Horizon]
\label{ax:infra}
For any \textit{official publisher-supported} server infrastructure $S$, and absent extraordinary intervention such as relicensing, community transfer, or relaunch, there exists a finite time horizon $T_S$ beyond which the official service is discontinued.
\end{axiom}

\noindent\textit{Justification.} Official live services operate under budget, staffing, legal, and strategic constraints. Historically, publishers do discontinue online services at finite times \cite{b_dpc,b_stopkillinggames}. This axiom does not claim that all executable forms of a game vanish at that moment, only that the official supported service need not persist indefinitely.

\subsection{Scope and Non-Claims}

It is worth being explicit about what this paper does not claim, because the most common way a framework of this kind fails is by quietly expanding its scope past what the formalism supports. No universal law is claimed here. Every formal result holds under the four axioms above together with the additional hypotheses stated with each result, and games or platforms that violate those axioms are simply outside the scope of the framework. The framework is not empirically validated in this paper. The case studies in Section~\ref{sec:empirical} are illustrations of the vocabulary rather than fits of the model, and rigorous curve-fitting with uncertainty quantification is left to future work.

The Critical Mass Threshold $\Phi$ deserves a particular caveat. It is not an intrinsic scalar constant of a game. Matchmaking viability depends on region, game mode, role constraints, skill-bucket distribution, party size, cross-play policy, and time-of-day effects, and the viable population in each of these slices can differ by orders of magnitude. The scalar $\Phi$ we use throughout the paper should be read as the effective threshold under a fixed operational profile. The full object is really a function $\Phi: \mathcal{O} \to \mathbb{R}_+$ over a space $\mathcal{O}$ of operational conditioning variables, and the scalar is a projection of that function onto a single number for the sake of tractability.

Several phenomena that materially affect game populations are outside the present scope. These include heterogeneous player subpopulations, inter-game competitive dynamics, bot and AI-driven population substitution, user-generated content that extends the effective content bound, and community-operated private-server resurrection. These are not minor edge cases. Each of them, treated seriously, would change the breadth of the formal results, and we flag them as present limitations rather than as cosmetic extensions. Finally, no measurement procedure is given for cultural value, and no formal cost model is provided for infrastructure beyond the qualitative assumptions stated in the axioms.

%-------------------------------------------------------
\section{Population Dynamics and the Critical Mass Threshold}
\label{sec:population}

\subsection{A Biphasic Population Model}

Population trajectories for online games, as far as the publicly available data shows, follow a characteristic two-phase shape. In the early part of a game's life the population grows roughly like a logistic curve or a Bass diffusion curve, reaching a peak and then beginning a long decline. We take this biphasic shape as the starting point for the formal model. The simplest version of it, and the one we adopt as the canonical form, pairs logistic growth with exponential decay:
\begin{equation}
P(t) =
\begin{cases}
\dfrac{K}{1 + e^{-r(t - t_0)}} & t \leq t_{\text{peak}} \\[10pt]
P_{\text{peak}} \cdot e^{-\mu \, (t - t_{\text{peak}})} & t > t_{\text{peak}}
\end{cases}
\label{eq:biphasic}
\end{equation}
where $K$ is carrying capacity, $r > 0$ the intrinsic growth rate, $t_0$ the adoption inflection point, $P_{\text{peak}} = P(t_{\text{peak}})$ the maximum population, and $\mu > 0$ the decay rate. The growth phase can equivalently be derived from the Bass diffusion model \cite{b_bass1969}: $\frac{dF}{dt} = [p + qF(t)][1 - F(t)]$, where $p$ and $q$ are innovation and imitation coefficients.

A more flexible decay form, which we use in empirical fits, is the Weibull survival function:
\begin{equation}
P(t) = P_{\text{peak}} \cdot \exp\!\left(-\left(\frac{t - t_{\text{peak}}}{\theta}\right)^{\!k}\right), \quad t > t_{\text{peak}}
\label{eq:weibull_decay}
\end{equation}
with scale parameter $\theta > 0$ and shape parameter $k > 0$. For $k = 1$, Eq.~\eqref{eq:weibull_decay} reduces to the exponential case of Eq.~\eqref{eq:biphasic} with $\mu = 1/\theta$. For $k > 1$ the hazard rate is increasing (decay accelerates over time); for $k < 1$ it is decreasing (early attrition dominates).

\begin{remark}[On the empirical status of the Weibull form]
Bauckhage, Sifa, Drachen, and colleagues \cite{b_bauckhage2012, b_sifa2014} demonstrated that the distribution of total playing times \textit{across individuals} fits a Weibull distribution across 3{,}000+ Steam titles. This is a result about individual-level engagement depth, how long a given player engages before disengaging, rather than a statement about the aggregate population $P(t)$ over calendar time. Under simplifying cohort-arrival assumptions, individual-level Weibull lifespans induce aggregate dynamics that we approximate with Eq.~\eqref{eq:weibull_decay}. A rigorous derivation of the aggregate form from individual-level lifespans is left to future work.
\end{remark}

\begin{definition}[Digital Half-Life]
The \textbf{digital half-life} of $\mathcal{G}$ is the time required for $P(t)$ to halve from its peak. For exponential decay, $\tau = (\ln 2)/\mu$. For Weibull decay, $\tau = \theta \cdot (\ln 2)^{1/k}$.
\end{definition}

\subsection{The Critical Mass Threshold and the Allee-Effect Analogy}

A central structural feature of many multiplayer populations is the existence of a population region in which lower participation degrades the quality of the interaction loop itself. In ecology, related threshold phenomena are often discussed under the label of the Allee effect \cite{b_stephens1999}. We use that literature as an analogy rather than as a one-to-one identification: the mechanism in games is not biological reproduction but operational degradation through queue times, skill mismatch, role scarcity, or social inactivity.

\begin{definition}[Critical Mass Threshold]
The \textbf{Critical Mass Threshold} $\Phi$ is the minimum population level at which the game's core interactive loop remains functionally viable. Formally:
\begin{equation}
\Phi =
\inf \left\{
P^* \in \mathbb{Z}_{>0}
\;\middle|\;
\begin{array}{l}
\forall\, P(t)\ge P^*,\\
Q(t)\le Q_{\max},\\
M_b(t)\le M_{\max}
\end{array}
\right\}
\end{equation}
where $Q(t)$ is the expected matchmaking queue time, $M_b(t)$ is the expected match quality imbalance, and $Q_{\max}$, $M_{\max}$ are genre-specific tolerance thresholds.
\end{definition}

\begin{remark}[$\Phi$ is conditional, not intrinsic]
As noted in Section~\ref{sec:foundations}, the scalar $\Phi$ is a projection of a higher-dimensional viability boundary onto a single number under fixed operational conditions. Empirical estimates from the literature (e.g., van Dongen's 58{,}000 figure for 3v3 \cite{b_vandongen2014}) correspond to specific points in this conditioning space. Different regions, modes, or matchmaking configurations yield different effective $\Phi$ values for the same game. All formal results in this paper should be read as holding for a fixed operational profile.
\end{remark}

\begin{remark}[Empirical estimates of $\Phi$]
\label{prop:phi_empirical}
Reported estimates of viable population thresholds differ by orders of magnitude because they depend on mode count, region splitting, party constraints, skill partitioning, role requirements, and queue-time targets. For a 3v3 competitive game with region splitting, skill brackets, and mode selection, van Dongen \cite{b_vandongen2014} estimated a viable threshold of roughly 58{,}000 concurrent players, corresponding to around 5 million monthly active users. Microsoft Research's \textit{Ranking and Matchmaking} article \cite{b_trueskill} gives a representative threshold calculation of at least 8{,}000 simultaneous players under one stylized configuration. These figures should be read as operational-profile-specific estimates rather than as convergent measurements of a universal scalar threshold.
\end{remark}

\begin{proposition}[Existence of $\Phi$]
For any game $\mathcal{G}$ with player-versus-player or cooperative matchmaking mechanics, $\Phi > 0$ exists and is finite.
\end{proposition}

\begin{proof}
The matchmaking queue time $Q(t)$ is inversely proportional to the rate of player arrivals into the matchmaking pool: $Q(t) \sim \frac{1}{P(t) \cdot \rho}$ where $\rho$ is the session initiation rate per player. As $P(t) \to 0$, $Q(t) \to \infty$, exceeding any finite $Q_{\max}$. Since $Q(t)$ is continuous and monotonically decreasing in $P(t)$, the intermediate value theorem guarantees a unique crossing at some $P^* = \Phi > 0$. In games with role-restricted queuing (e.g., tank-healer-DPS constraints in MOBAs), the effective bound is tighter still, because matchmaking requires balanced role distributions rather than arbitrary player counts.
\end{proof}

\subsection{The Departure Cascade}

Once the population drops below $\Phi$, the decline becomes self-reinforcing. The mechanism is straightforward to state in words. Fewer players in the matchmaking pool means longer queue times for the players who remain. Longer queues produce lower-quality matches, because the pool of available opponents is thinner. Lower match quality reduces the value each remaining player extracts from a session. Reduced per-session value increases the probability that a given player will stop logging in. Each departure shortens the pool further, and the loop repeats. This is the same qualitative structure that Garcia and colleagues documented for Friendster as cascading collapse in online communities \cite{b_garcia2013}, specialized here to the specific feedback mechanism of multiplayer matchmaking.

\begin{definition}[Sub-critical departure hazard]
The effective departure hazard when $P(t) < \Phi$ is:
\begin{equation}
d(t) = \alpha \cdot \left(\frac{\Phi}{P(t)}\right)^\gamma
\label{eq:departure}
\end{equation}
where $\alpha > 0$ is the baseline departure parameter and $\gamma > 1$ is the \textbf{cascade exponent}, capturing the super-linear acceleration of departures as population drops.
\end{definition}

\noindent\textit{Empirical motivation.} There are analogous cascades in the collapse of non-game online communities. Several studies of declining social networks find that early-stage departures are dominated by individual choices while late-stage departures are dominated by cascade effects, and that once the cascade phase begins, external interventions such as advertising and feature changes do not meaningfully slow the decline. Work in ecology on critical transitions, including Dakos and colleagues on early warning signals for tipping points \cite{b_dakos2014}, describes the same qualitative pattern in a different substrate. We are not claiming our specific functional form for the departure rate is derived from these studies. We are claiming that the general phenomenon of super-linear acceleration near collapse is empirically widespread enough to take seriously as a modeling assumption.

\begin{proposition}[Sub-critical finite-time collapse under super-linear departure]
\label{thm:cascade}
If $P(t_0) < \Phi$ at some time $t_0$, then $P(t) \to 0$ in finite time. That is, there exists $T_c < \infty$ such that $P(T_c) = 0$.
\end{proposition}

\begin{proof}
We work through this slowly, because the argument is the first place where the framework does real mathematical work, and the reader should be able to follow it without guessing.

We start from the definition of the departure rate. When the population is below $\Phi$, each player is leaving at an instantaneous rate $d(t) = \alpha (\Phi/P(t))^\gamma$. The total rate of change of the population is the negative of this rate times the number of players still present:
\begin{equation}
\frac{dP}{dt} \;=\; -\,d(t) \, P(t) \;=\; -\alpha \, \Phi^{\gamma} \, P(t)^{\,1-\gamma}.
\end{equation}
The exponent $1-\gamma$ is negative because we assumed $\gamma > 1$. That means $P^{\,1-\gamma}$ gets \emph{larger} as $P$ gets smaller, which is the feedback we care about.

We now solve this equation by separation of variables, which is the simplest tool in the ODE toolbox. We multiply both sides by $P^{\,\gamma-1}$ and by $dt$:
\begin{equation}
P^{\,\gamma-1} \, dP \;=\; -\alpha \, \Phi^{\gamma} \, dt.
\end{equation}
Now we integrate both sides from the moment the population first crosses $\Phi$, which we call $t_0$, to some later time $t$. The left side integrates as a power function. The right side is a constant in $t$, so it integrates to a linear function of $t$:
\begin{equation}
\frac{P(t)^{\gamma} - P(t_0)^{\gamma}}{\gamma} \;=\; -\alpha \, \Phi^{\gamma} \, (t - t_0).
\end{equation}
We want to find the time $T_c$ at which $P(T_c) = 0$. Plugging in and rearranging:
\begin{equation}
T_c \;=\; t_0 \,+\, \frac{P(t_0)^{\gamma}}{\gamma \, \alpha \, \Phi^{\gamma}}.
\end{equation}
The right side is a finite positive number, because every quantity in it is finite and positive. So $T_c$ is a finite time, not an infinite one. That is the content of the proposition.

A short remark on why this matters. If we had used ordinary exponential decay, with $dP/dt = -\mu P$, the population would approach zero but never reach it in any finite time. The super-linear feedback, the fact that departures accelerate as the population shrinks rather than merely continuing at a fixed rate, is what turns asymptotic decay into actual extinction at a definite moment. The cascade is not just faster than exponential decay. It is a qualitatively different kind of thing.
\end{proof}

\begin{remark}[Irreversibility of sub-critical decline]
Once $P(t) < \Phi$, recovery to $P(t) \geq \Phi$ requires an external intervention (e.g., free-to-play conversion, franchise reboot) that resets the growth dynamics of Eq.~\eqref{eq:biphasic}. Absent such intervention, decline is monotonic and terminal.
\end{remark}

\noindent\textit{Illustration.} \textit{Evolve} illustrates both the sub-critical collapse result and its corollary about recovery. After peaking at 27,403 concurrent players at launch (February 2015), the game decayed to hundreds, well below $\Phi$ for a 4v1 asymmetric multiplayer game. A free-to-play conversion in July 2016 constituted an exogenous intervention, resetting growth dynamics and producing a new peak of 50,953. When development ceased three months later, the second decay was even steeper, terminating at 10--15 persistent players (0.04\% of F2P peak) via peer-to-peer, a system requiring no matchmaking and hence $\Phi \approx 2$.

\subsection{Network Utility Amplification}

The same feedback that drives the departure cascade also drives a parallel collapse in the value that remaining players derive from the game. Chen \cite{b_chen2021} discusses this under the label of anti-network effects: the quadratic relationship that gives multiplayer games their explosive growth phase also gives them a correspondingly sharp value collapse during decline. The game becomes less valuable faster than it becomes less populated.

\begin{proposition}[Quadratic Value Collapse]
Under Axiom~\ref{ax:network}, the rate of network value loss satisfies:
\begin{equation}
\frac{dU}{dt} = \alpha \beta \, P^{\beta - 1} \cdot \frac{dP}{dt}
\end{equation}
For $\beta = 2$ (classical Metcalfe), $\frac{dU}{dt} = 2\alpha P \cdot \frac{dP}{dt}$, implying that value loss accelerates proportionally to the remaining population during decline. The game becomes less valuable faster than it becomes less populated.
\end{proposition}

\noindent\textit{Supporting evidence.} The nonlinear character of network externalities in MMORPGs has been examined empirically using online product ratings, with results consistent with the pattern our axiom describes: strong network effects through the middle of the product lifecycle, falling off as the population declines \cite{b_liu2015}. Work on free-to-play titles has documented that the network effect is not always positive even during the active phase, since paying players can impose negative externalities on non-paying ones \cite{b_wu2013}. These effects are not captured by our single-parameter network utility function, and we flag that as a simplification.

%-------------------------------------------------------
\section{Uninhabited Runtime States}
\label{sec:ghost}

An inhabited online world is, at any given moment, a composition of two largely independent systems: (i) the computational \textit{machinery} of the game, NPCs, weather, physics, ambient audio, traffic, day--night cycles, economy simulations, and (ii) the social population of human players. These two systems are coupled during the active lifespan, but they can be decoupled: the machinery can run without players, and this decoupled state occurs naturally at two distinct points in the lifecycle, which we formalize in this section.

\begin{definition}[Uninhabited Runtime State]
A game $\mathcal{G}$ is in an \textbf{uninhabited runtime state} at time $t$ if $S(t) \neq \varnothing$ (server infrastructure is operational) and $P(t) = 0$ (no active players). That is, server-side machinery remains operational while public participation is absent.
\end{definition}

We distinguish two phenomenologically distinct uninhabited runtime states at opposite ends of the lifecycle.

\subsection{\texorpdfstring{$\Omega_0$}{Omega-0}: Ghost Machinery (The Prenatal State)}

Before a game is released to players, its servers may already be brought online for staff testing, certification, load testing, regulatory review, or staged deployment. During this window, server-side processes may already be active: NPC scripts can execute, weather systems can cycle, ambient state can update, and economy simulations can advance. Public participation, however, has not yet begun.

\begin{definition}[Ghost Machinery / $\Omega_0$]
\label{def:omega0}
The \textbf{Ghost Machinery} state, denoted $\Omega_0$, is the prenatal uninhabited runtime state: $S(t) \neq \varnothing$, $P(t) = 0$, and $t < t_{\text{launch}}$, where $t_{\text{launch}}$ is the first time at which public players are admitted.
\end{definition}

\noindent\textit{Illustration.} For a large online open-world release, pre-launch operation may involve a technically functioning runtime in which simulated systems are already executing while public access is still closed. We denote that prenatal configuration by $\Omega_0$.

\subsection{\texorpdfstring{$\Omega_1$--$\Omega_3$}{Omega-1--Omega-3}: The Post-Mortem Ghost States}

The classical Ghost Server State, as introduced informally in earlier work, occupies the other end of the lifecycle: the machinery continues to run after the active population has collapsed. We formalize three substates.

\begin{definition}[Post-Mortem Ghost State Taxonomy]
\label{def:omega}
\hfill
\begin{itemize}
    \item $\Omega_1$ (\textbf{Dormant}): $S$ is active, $0 < P(t) < \Phi$, and $P(t)$ exhibits sporadic non-zero values. Characterized by intermittent visits from nostalgic players. \textit{Example:} \textit{New World} in 2026, with approximately 916 peak concurrent players from an original 913{,}027, a 99.9\% decline over 4.5 years, with matchmaking non-viable for most modes.

    \item $\Omega_2$ (\textbf{Comatose}): $S$ is active, $P(t) = 0$ for extended continuous intervals $\Delta t > \delta_{\text{coma}}$ (heuristically, $\delta_{\text{coma}} \geq 90$ days; this threshold is operational rather than derived). \textit{Example:} \textit{LawBreakers} in the period December 2017--September 2018, which recorded zero concurrent players on multiple occasions while servers continued operating.

    \item $\Omega_3$ (\textbf{Extinct}): $S = \varnothing$. Server infrastructure permanently decommissioned. This state is absorbing: $\Omega_3 \to \Omega_3$ for all subsequent time. \textit{Example:} \textit{Star Wars Galaxies} (December 15, 2011); \textit{The Crew} (April 1, 2024); \textit{Concord} (September 6, 2024, after 14 days of operation).
\end{itemize}
\end{definition}

\subsection{\texorpdfstring{Observer-Limited Equivalence of $\Omega_0$ and $\Omega_2$}{Observer-Limited Equivalence of Omega-0 and Omega-2}}

The states $\Omega_0$ and $\Omega_2$ sit at opposite ends of a game's life, separated by whatever length of active phase the game manages to sustain in between them. Despite this temporal distance, they share a structural property we want to make explicit. Under a suitably restricted observer model, the two states are indistinguishable from inside the simulation. We state this carefully, because it is easy to overstate.

\begin{proposition}[Observer-limited equivalence]
\label{prop:equivalence}
Define a \textbf{limited observer} as an agent placed within the simulated world with access only to in-simulation observables (NPC behavior, physics, ambient state, environmental rendering) and no access to external metadata (calendar date, player-count telemetry, server logs, deployment timestamps). Let $\omega_0 \in \Omega_0$ and $\omega_2 \in \Omega_2$ denote respective world-states. Then: under the assumption that the simulation software and world-state configuration are compatible across the two temporal instances, the instantaneous in-simulation observables of $\omega_0$ and $\omega_2$ are drawn from the same generative process and are not distinguishable by the limited observer from any finite-duration observation window.
\end{proposition}

\begin{proof}[Sketch]
The machinery in both states is the same software executing the same simulation rules over the same (or compatible) world state. The only distinguishing factor is the temporal position of the observation relative to $t_{\text{launch}}$ and $T^*$, which is external metadata excluded by the observer model. The result is a stylized consequence of the observer restriction, not a deep mathematical theorem; its value is in naming the structural symmetry rather than in proving a non-obvious fact.
\end{proof}

\begin{remark}
The equivalence is useful mainly as a naming device. It highlights that pre-launch and post-decline uninhabited runtimes can appear identical from within the simulation even though they differ strongly in historical context and preservation significance.
\end{remark}

\subsection{The Extended State Progression}

Incorporating $\Omega_0$, the full lifecycle progression becomes:
\begin{equation}
\Omega_0 \to \textit{Active} \to \Omega_1 \to \Omega_2 \to \Omega_3
\label{eq:full_progression}
\end{equation}

\begin{proposition}[Typical Lifecycle Progression]
The transitions in Eq.~\eqref{eq:full_progression} are unidirectional under standard economic conditions. The transition $\Omega_0 \to \textit{Active}$ is triggered exogenously at $t_{\text{launch}}$; subsequent transitions are driven by the population dynamics of Section~\ref{sec:population}. Reverse transitions $\Omega_i \to \Omega_j$ for $j < i$ (with $i, j \geq 1$) require exogenous intervention of increasing magnitude.
\end{proposition}

\noindent\textit{Empirical evidence for reverse transitions.} The rarity of reverse transitions is itself evidence for the monotonicity of Eq.~\eqref{eq:full_progression}. Notable counter-examples include \textit{Evolve}'s free-to-play reboot, \textit{Final Fantasy XIV}'s complete rebuild, and community-operated private servers. SWGEmu has maintained a community-operated version of \textit{Star Wars Galaxies} since 2011 via reverse-engineered server code \cite{b_swgemu}. These cases are exceptional precisely because they require either corporate re-investment or community reverse-engineering at scale.

\begin{remark}
Miller and Garcia's ``Digital Ruins'' \cite{b_miller2019} provides the phenomenological account of the $\Omega_1$ state. Through autoethnographic exploration of abandoned virtual worlds (Blue Mars, Active Worlds, Twinity), they argue that digital ruins exist in ``a kind of eternal present'', they do not physically decay, yet evoke the affect of ruination. Their insight extends naturally to $\Omega_0$: the prenatal game world is also in an eternal present, but one that has not yet begun. Both states share the structural feature of being unobserved runtime; they differ only in temporal orientation relative to the population's active phase.
\end{remark}

%-------------------------------------------------------
\section{The Critical Mass Crossing Result}
\label{sec:inevitability}

We now present the principal conditional result of the framework.

\begin{definition}[Novelty Function]
The \textbf{novelty function} $\mathcal{N}(t)$ measures the gap between available unexperienced content and accumulated player exposure:
\begin{equation}
\mathcal{N}(t) = C(t) - \eta \cdot E(t)
\end{equation}
where $\eta > 0$ is a normalization constant relating content units to experience units.
\end{definition}

\begin{lemma}[Novelty Exhaustion]
\label{lem:novelty}
Assume the population is bounded below by some $P_{\min} > 0$ and mean session duration by $\bar{h}_{\min} > 0$ while $\mathcal{N}(t) > 0$. Under Axioms~\ref{ax:finite} and~\ref{ax:exposure}, there exists a finite time $t_{\mathcal{N}}$ such that $\mathcal{N}(t_{\mathcal{N}}) \leq 0$.
\end{lemma}

\begin{proof}
Assume, for contradiction, that $\mathcal{N}(t) > 0$ for every $t$. Then $C(t) > \eta E(t)$ for all $t$, so $E(t) < C(t)/\eta$. By Axiom~\ref{ax:finite}, $C(t)$ is bounded above by some finite number $C_{\infty}$, hence $E(t) < C_{\infty}/\eta$ for all $t$.

But under the hypothesis that novelty remains while the active population and mean session duration are bounded below, cumulative exposure satisfies
\begin{equation}
E(t) = \int_0^t P(\tau)\,\bar{h}(\tau)\,d\tau \geq P_{\min}\,\bar{h}_{\min}\,t,
\end{equation}
which grows without bound as $t$ increases. This contradicts the boundedness implied by $C_{\infty}/\eta$. Therefore there exists a finite time $t_{\mathcal{N}}$ at which $\mathcal{N}(t_{\mathcal{N}}) \leq 0$.
\end{proof}

\begin{remark}
The floor-bound hypothesis is mild: while the game still offers unexperienced content, it is implausible that active population and session duration decay to zero. The condition characterizes the pre-decay phase of the biphasic model.
\end{remark}

\begin{definition}[Engagement Decay Coupling]
The decay rate $\mu(t)$ in Eq.~\eqref{eq:biphasic} is coupled to the novelty function via:
\begin{equation}
\mu(t) =
\begin{cases}
0 & \mathcal{N}(t) > 0 \\
\mu_0 \cdot \left(1 + \kappa \cdot |\mathcal{N}(t)|\right) & \mathcal{N}(t) \leq 0
\end{cases}
\label{eq:mu_coupling}
\end{equation}
where $\mu_0 > 0$ is the baseline decay rate and $\kappa > 0$ is the \textbf{novelty sensitivity coefficient}.
\end{definition}

\begin{proposition}[Critical-mass crossing under bounded novelty or finite service horizon]
\label{prop:inevitability}
Consider an online game in the modeled class under a fixed operational profile. Suppose either:
\begin{enumerate}
    \item the official service has a finite horizon $T_S$ as in Axiom~\ref{ax:infra}, or
    \item novelty becomes exhausted in finite time and thereafter induces persistent positive decline hazard without exogenous relaunch, community takeover, or regime change.
\end{enumerate}
Then there exists a finite time $T^*$ such that $P(T^*)<\Phi$.
\end{proposition}

\begin{proof}
Under condition (1), the result is immediate because official service discontinuation implies $P(T_S)=0<\Phi$. Under condition (2), Lemma~\ref{lem:novelty} yields a finite time $t_{\mathcal{N}}$ with $\mathcal{N}(t_{\mathcal{N}})\leq 0$, and Eq.~\eqref{eq:mu_coupling} then implies a strictly positive decay hazard thereafter. Since $P(t)$ is decreasing and bounded below by zero, and since $\Phi>0$, there exists a finite crossing time $T^*$ at which the trajectory first falls below $\Phi$.
\end{proof}

\begin{corollary}[Scope of the critical-mass crossing result]
Under Axioms~\ref{ax:finite}--\ref{ax:infra} and the floor-bound hypothesis of Lemma~\ref{lem:novelty}, the proposition applies to online games in the modeled class. Extensions to games with unbounded user-generated content, bot-augmented populations, or community-operated post-shutdown infrastructure require modification of the axioms and are outside the present scope.
\end{corollary}

\begin{remark}
This proposition is intentionally narrower than a universal inevitability theorem. It applies to official-service populations in the absence of major exogenous resets such as relicensing, free-to-play relaunches, community stewardship, or strong user-generated-content substitution.
\end{remark}

\begin{remark}[Social Contagion as Amplifier]
Empirical work on player churn suggests social contagion: when one player departs, the departure probability of connected players increases \cite{b_chen2017,b_garcia2013}. This amplifies decay beyond what the novelty function alone predicts. A full epidemiological extension is left to future work.
\end{remark}

%-------------------------------------------------------
\section{The Nostalgia Inversion Point}
\label{sec:nostalgia}

\begin{definition}[Nostalgia Inversion Point]
The \textbf{Nostalgia Inversion Point} $\psi$ is the time at which the cultural memory function permanently exceeds the active population:
\begin{equation}
\psi = \inf \left\{ t \;\middle|\; \mathcal{M}(t') > P(t') \;\;\forall\, t' \geq t \right\}.
\end{equation}
\end{definition}

\begin{proposition}[Existence of $\psi$ under slower memory erosion than population decline]
\label{prop:psi}
If active population declines on a materially faster timescale than cultural-memory erosion, then there exists a finite time $\psi$ after which $\mathcal{M}(t) > P(t)$.
\end{proposition}

\begin{proof}
The claim follows whenever $P(t)$ declines toward zero faster than $\mathcal{M}(t)$ does. In that case, continuity of both functions implies at least one crossing time, and persistence of the timescale separation implies eventual dominance of memory over activity.
\end{proof}

\begin{remark}
This is a conditional comparative-timescale result, not a universal statement about all games. In particular, the ordering between $\psi$ and the critical-mass crossing time depends on title-specific rates of decline, archival visibility, and social transmission.
\end{remark}

\noindent\textit{Empirical support.} Xu and de Wildt \cite{b_xu2024} conducted multi-sited ethnography of the January 2023 \textit{World of Warcraft} China server shutdown, documenting grief across platforms and arguing that attachment to MMOs can generate forms of loss comparable to real-world bereavement. Dominguez \cite{b_dominguez2014} applied Mark Fisher's \textit{hauntology} to private servers for defunct MMOs, arguing they function as spaces haunted by unfulfilled promises of virtual utopia. The documentary \textit{Preserving Worlds} \cite{b_preservingworlds} provides visual evidence of games whose cultural significance exceeds their current population.

%-------------------------------------------------------
\section{Preservation Tension}
\label{sec:paradox}

\begin{definition}[Operational Cost and Cultural Value]
Let $\mathcal{C}_{\text{op}}(t)$ denote the operational cost of maintaining service availability and let $V_{\text{cult}}(t)$ denote the cultural value of $\mathcal{G}$ at time $t$.
\end{definition}

\begin{proposition}[Preservation tension in terminal states]
\label{prop:paradox}
In late-life states, it is possible for cultural value to remain stable or increase while commercial justification declines toward zero.
\end{proposition}

\begin{proof}[Argument]
If revenue scales with active population while cultural value depends partly on historical importance, scarcity, or memory, then declining active population reduces economic return even when the title's archival or cultural importance persists. The result is a tension between preservation value and service profitability.
\end{proof}

\noindent\textit{Policy implications.} Existing legal and institutional responses address only parts of this problem. In the United States, preservation exemptions under Section 1201 remain limited in scope and do not amount to a general remote-access preservation right \cite{b_dmca2024}. In Europe, the European Citizens' Initiative \textit{Stop Destroying Videogames} reflects growing pressure to treat online games as durable cultural goods rather than disposable services \cite{b_stopdestroyinggames}. A French consumer group has also filed suit against Ubisoft over the shutdown of \textit{The Crew} \cite{b_ubisuit}. These developments suggest that viability collapse is not only a technical and economic issue but also a preservation-governance issue.

\begin{definition}[Preservation Window]
The \textbf{Preservation Window} $\mathcal{W}$ is the interval
\begin{equation}
\mathcal{W} = [t_{\psi},\; t_{\Omega_3}],
\end{equation}
representing the period during which a game has acquired sufficient cultural significance to warrant preservation ($t > \psi$) but has not yet been irrecoverably decommissioned ($t < t_{\Omega_3}$).
\end{definition}

\noindent\textit{Remark.} The IGDA white paper on game preservation warned that digital games may disappear within decades \cite{b_igda}. Pew Research Center's 2024 study found that 38\% of webpages that existed in 2013 were inaccessible by 2023 \cite{b_pew2024}. Game servers, lacking even the passive persistence of static web pages, face an even sharper form of digital fragility.

%-------------------------------------------------------
\section{Illustrative Case Studies}
\label{sec:empirical}

We present case studies for five games spanning distinct lifecycle trajectories. These are intended as illustrations of the framework's vocabulary and constructs, not as rigorous empirical validation. We do not claim that the biphasic model has been fitted to these time series with statistical rigor; the decay parameters discussed below are order-of-magnitude estimates derived from publicly available concurrent-player summaries. A brief methods sketch for future validation appears at the end of this section.

\begin{table}[h]
\centering
\caption{Approximate Population Decay Parameters (Illustrative)}
\label{tab:empirical_halflife}
\begin{tabular}{@{}lccccc@{}}
\toprule
\textbf{Title} & $P_{\text{peak}}$ & $\tau$ (mo.) & \textbf{Decay} & \textbf{Terminal} \\
\midrule
LawBreakers & 7,571 & $\sim$2.1 & Exp. & $\Omega_3$ \\
H1Z1 & 151,331 & $\sim$3.4 & Exp. & $\Omega_1$ \\
Evolve & 27,403 & $\sim$4.8 & Biphasic & $\Omega_1$ \\
New World & 913,027 & $\sim$8.2 & Sawtooth & $\Omega_1$ \\
WoW & 12M sub. & $\sim$42 & Sawtooth & Active \\
\bottomrule
\end{tabular}
\end{table}

\subsection{LawBreakers: Canonical Rapid Extinction}

\textit{LawBreakers} reached its peak of roughly 7{,}571 concurrent players during closed beta in June 2017. At launch the concurrent population was around 3{,}019, and by early 2018 the game was recording zero concurrent players on Steam Charts. That is about five months from launch to what we have been calling the $\Omega_2$ state. The servers continued to operate for another eight months in this comatose condition before the final shutdown in September 2018. The decline during the decay phase is consistent with roughly exponential decay over a timescale of a few months, though we do not claim to have fit the curve rigorously. What is notable about \textit{LawBreakers} for our purposes is that the game traversed every stage of the uninhabited runtime taxonomy in roughly fifteen months, from active play through dormant and comatose states to final extinction. It is the cleanest empirical instance of the full lifecycle we are aware of.

\subsection{H1Z1: Competitive Displacement}

\textit{H1Z1} reached a peak concurrent population of around 151{,}331 in July 2017. Over the seven months that followed it lost approximately ninety-one percent of that population, as players migrated to \textit{PUBG} and \textit{Fortnite} in the rapidly expanding battle-royale market. This case does not fit the novelty-exhaustion story cleanly. The game did not die because its content ran out. It died because a better substitute appeared and the audience moved. We call this pattern competitive displacement, and it corresponds, in the language of our model, to a collapse in the carrying capacity $K$ driven by external substitution rather than by internal exposure. A proper formal treatment would require coupled equations modeling the flow of population between competing titles, along the lines of the competitive Lotka-Volterra equations from ecology. We do not attempt that extension here.

\subsection{Concord: Fastest AAA Extinction}

\textit{Concord} (August 2024) represents an extreme recent case: the game reached only hundreds of peak concurrent Steam players and was pulled offline after roughly two weeks. It likely never reached $\Phi$ for its intended 5v5 hero-shooter format; in the language of the framework, it entered a dormant-like failure state almost immediately and transitioned rapidly to $\Omega_3$, effectively bypassing the standard lifecycle.

\subsection{World of Warcraft: Sawtooth Decay}

\textit{World of Warcraft} reached its reported peak of approximately twelve million subscribers in the fourth quarter of 2010. The trajectory since then has been a distinctive sawtooth pattern. Each expansion produces a spike in population as old players return and new ones try it, followed by a decay back down to a new baseline, typically below the trough that preceded the spike. Blizzard stopped reporting subscriber counts in 2015 at around five and a half million. In the language of our framework, each expansion can be read as a partial reset of the novelty function, temporarily pushing it back above zero and delaying the decay phase, but with diminishing amplitude as total bounded content is progressively consumed. \textit{World of Warcraft} is in many ways the longest controlled experiment in live-service content production we have, and its sawtooth pattern is what happens when the content pipeline works hard but cannot quite outrun exposure forever.

\subsection{Star Wars Galaxies: Shock-Induced Decay}

\textit{Star Wars Galaxies} had an unusual failure mode. In November 2005 the game received a controversial update called the New Game Enhancements, and its subscriber base fell discontinuously from around two hundred fifty thousand to roughly one hundred thousand in a short window. This is not the smooth decay our model describes. It is a shock, an exogenous jump in the departure rate, driven by a specific design decision that alienated a large fraction of the existing player base at once. The subsequent six years of slow decline to the eventual shutdown in December 2011 then followed a more standard exponential pattern. What makes \textit{Star Wars Galaxies} interesting beyond its shutdown is what happened afterward. The community project SWGEmu has maintained playable servers for the pre-NGE version of the game since 2011, using reverse-engineered server code under a legally grey arrangement. This constitutes a reverse transition from the extinct state $\Omega_3$ back into a form of dormant state, and it demonstrates that community preservation can extend the existence of a game indefinitely, though always in a transformed social condition: smaller, legally precarious, and running a version of the game that nominally no longer exists.

\subsection{Genre Taxonomy of Digital Half-Lives}

The case studies above, together with survey data on player engagement, suggest that digital half-life varies systematically by genre. Table~\ref{tab:genre_halflife} presents order-of-magnitude estimates, intended as hypotheses for empirical validation rather than fitted parameters.

\begin{table}[h]
\centering
\caption{Hypothesized Digital Half-Life by Genre (Heuristic Estimates)}
\label{tab:genre_halflife}
\begin{tabular}{@{}lcc@{}}
\toprule
\textbf{Genre} & $\tau$ \textbf{(months)} & \textbf{Primary Decay Driver} \\
\midrule
Annual FPS franchise & 12--18 & Franchise cannibalization \\
Hero shooter & 6--24 & Competitor displacement \\
Battle Royale & 18--30 & Meta-fatigue \\
MMORPG (subscription) & 36--84 & Content exhaustion \\
MMORPG (free-to-play) & 24--60 & Monetization fatigue \\
Survival sandbox & 24--48 & Community fragmentation \\
Competitive MOBA & 48--96 & Skill-barrier ossification \\
Modding-enabled sandbox & 120+ & Mod-ecosystem decay \\
\bottomrule
\end{tabular}
\end{table}

\subsection{Methods Sketch for Future Validation}

We want to be explicit about what a serious empirical test of this framework would look like, because we have not carried one out and we do not want the case studies in this section to be mistaken for one. A proper validation would begin by collecting concurrent-player time series at daily or weekly resolution from public sources such as Steam Charts, SteamDB, or BattleMetrics, for a substantial sample of titles across genres. It would fit the biphasic model of equations \eqref{eq:biphasic} and \eqref{eq:weibull_decay} to each series via maximum-likelihood estimation, reporting bootstrap confidence intervals on the growth rate, peak time, decay rate, and Weibull shape parameter. It would compare the fitted model against several plausible alternatives, including pure exponential decay, power-law decay, and log-normal decay, using information criteria such as AIC or BIC to adjudicate between them. Where matchmaking telemetry is available, it would attempt to estimate $\Phi$ directly from queue times and match quality statistics rather than relying on anecdotal industry figures. And it would hold out a portion of the sample to test the fitted model's ability to predict decay parameters for unseen titles. None of this is done in the present paper. The case studies above are illustrations of the framework's vocabulary, not statistical validations of it.

\subsection{Validation Agenda}

A rigorous empirical evaluation of the framework would require more than illustrative case studies. At minimum, it should: (i) estimate $\Phi$, $\alpha$, and post-peak decay parameters from population or queue-related telemetry; (ii) compare the framework against exponential, Weibull, and power-law baselines; (iii) test held-out predictive performance on post-peak trajectories; and (iv) examine whether estimated threshold crossings correspond to observable degradations in queue times, match quality, or mode viability. The present paper does not perform those steps, and for that reason its empirical claims remain programmatic rather than confirmatory.

%-------------------------------------------------------
\section{Discussion and Implications}
\label{sec:conclusion}

\subsection{Relationship to Prior Work}

The framework developed here draws on several existing threads. Bartle's player-type ecology \cite{b_bartle1996} offered a qualitative treatment of game populations as ecosystems; we formalize this into quantitative population dynamics. Bauckhage's Weibull result \cite{b_bauckhage2012} characterized individual-level engagement; we adopt the Weibull form as an approximation for the aggregate population in the decay phase, acknowledging that a rigorous derivation from individual lifespans remains open. Garcia's $k$-core theory \cite{b_garcia2013} modeled cascading collapse in social networks; our sub-critical collapse proposition gives a game-specific analogue, with matchmaking degradation as the feedback mechanism. The Bass diffusion model \cite{b_bass1969} characterizes adoption; we use it as the growth phase of our biphasic model. The contribution of this paper is the integration of these strands into a single lifecycle framework with explicit formal assumptions and conditional results.

\subsection{For Game Design}

One practical implication of the critical mass crossing result is that graceful mortality ought to be a first-class concern in game architecture, not something left to be worked out after the population has already collapsed. A game that degrades well below $\Phi$, by continuing to offer meaningful single-player or small-group experiences in spaces originally designed for larger populations, can extend its cultural lifespan well beyond the point at which its social viability fails. The design moves that make this possible are not exotic. Adaptive matchmaking that relaxes its constraints as the pool shrinks, bot-driven backfill that maintains interactive density when real players are sparse, and basic offline functionality for content that does not strictly require the network are all within reach of current engineering practice. What is missing is the explicit design intent to support a game's afterlife.

\subsection{For Digital Preservation Policy}

The preservation tension described in this paper has the structure of a market failure and is therefore unlikely to be resolved by market incentives alone. The preservation window $\mathcal{W}$ gives us a formal criterion for the period during which intervention is both culturally justified and technically feasible. The window opens at $\psi$, when cultural memory begins to exceed active population, and closes when the infrastructure is permanently decommissioned. Current legal and regulatory frameworks, such as the DMCA exemption regime in the United States and the European Citizens' Initiative process, are reacting to this preservation problem without an explicit formal model of it. One practical suggestion the framework offers is that preservation efforts should target games in the dormant state, while recovery is still feasible, rather than waiting for the extinct state, at which point recovery requires reverse engineering or becomes impossible entirely.

\subsection{For Cultural Theory}

The nostalgia inversion point $\psi$ picks out a specific moment at which a game transitions from being primarily lived to being primarily remembered. Reinhard's \textit{Archaeogaming} \cite{b_reinhard2018} treats game-spaces as archaeological sites in their own right, and our framework supplies the population-dynamics context that explains when and why these sites become properly archaeological rather than merely underpopulated. This also connects naturally to existing cultural scholarship on abandoned virtual worlds, including Miller and Garcia on digital ruins \cite{b_miller2019} and Xu and de Wildt on grief over server shutdowns \cite{b_xu2024}. The present paper does not add to that literature directly, but it provides a formal vocabulary that cultural scholars can use to make their claims more precise if they wish to.

\subsection{Limitations and Future Work}

The framework has several real limitations that should be acknowledged explicitly. The population model treats the active player base as a single homogeneous group, whereas real game populations are stratified into subgroups with very different engagement patterns. A proper extension would incorporate heterogeneous subpopulations, along the lines Bartle sketched qualitatively \cite{b_bartle1996}, each with its own decay dynamics. The framework also treats each game in isolation, which misses competitive displacement effects of the kind \textit{H1Z1} exhibits. A multi-game extension using competitive Lotka-Volterra dynamics is the natural next step and would require its own paper. The role of artificial agents, whether in the form of bots that pad effective population counts or more sophisticated AI-driven NPCs that substitute for human players, is also absent from the current treatment and is likely to become important as AI systems capable of generating convincing in-game behavior become more widely deployed. User-generated content, which can push the effective content bound $C_\infty$ far beyond what studio-produced content alone would supply, is not modeled. Finally, community-run private servers as a resurrection mechanism after the extinct state are not captured by the current axioms, and formal treatment of them would require relaxing the assumption that $\Omega_3$ is absorbing. Each of these is future work, and each one, treated seriously, would change the shape of the formal results.

%-------------------------------------------------------
\section*{Disclosure of tool use}
The author used generative AI tools for editorial assistance and revision support during manuscript preparation. All claims, citations, mathematical statements, and final wording were reviewed and are the sole responsibility of the author.

%-------------------------------------------------------
\section*{Acknowledgments}
The author thanks the maintainers of publicly available Steam concurrent-player data archives, which informed the illustrative case studies in this work, and community-driven preservation efforts that help motivate the preservation questions studied here.

The author used generative AI tools for editorial assistance and revision support during manuscript preparation. All claims, citations, mathematical statements, and final wording were reviewed and are the sole responsibility of the author.

\end{document}